\documentclass[pagenumbers, final]{leapaper}
\svnInfo $Id$

\usepackage{tikz}
\usetikzlibrary{snakes}
\usepgflibrary{shapes.arrows}

\title{The Efficiency of MapReduce in Parallel External Memory}
\author{Gero Greiner \and Riko Jacob}
\date{{Institute of Theoretical Computer Science \\ ETH Zurich, Switzerland\\ \texttt{\{greinerg,rjacob\}@inf.ethz.ch}}}
\institute{Institute of Theoretical Computer Science, ETH Zurich, Switzerland\\ \texttt{\{greinerg,rjacob\}@inf.ethz.ch}}

\makeatletter
\newif\if@fullpaper
\@fullpapertrue

\begin{document}

\maketitle
\begin{abstract}
 Since its introduction in 2004, the MapReduce framework has become one of the standard approaches in massive distributed and parallel computation.
 In contrast to its intensive use in practise, theoretical footing is still limited and only little work has been done yet to put MapReduce on a par with the major computational models.
 Following pioneer work that relates the MapReduce framework with PRAM and BSP in their macroscopic structure, we focus on the functionality provided by the framework itself, considered in the parallel external memory model (PEM).
 In this, we present upper and lower bounds on the parallel I/O-complexity that are matching up to constant factors for the shuffle step.
 The shuffle step is the single communication phase where all information of one MapReduce invocation gets transferred from map workers to reduce workers.
 Hence, we move the focus towards the internal communication step in contrast to previous work.
 The results we obtain further carry over to the BSP$^*$ model.
 On the one hand, this shows how much complexity can be ``hidden'' for an algorithm expressed in MapReduce compared to PEM.
 On the other hand, our results bound the worst-case performance loss of the MapReduce approach in terms of I/O-efficiency.  
\end{abstract}

\section{Introduction}
The MapReduce framework has been introduced by Dean and Ghemawat~\cite{MapReduce04} to provide a simple parallel model for the design of algorithms on huge data sets.
It allows an easy design of parallel programs that scale to large clusters of hundreds or thousands of PCs.
Since its introduction in 2004, apart from its intensive use by Google for tasks involving petabytes of data each day~\cite{MapReduce10}, the open source implementation Hadoop~\cite{MR_Hadoop} has found many applications including regular use by companies like Yahoo!, eBay, Facebook, Twitter and IBM. 
This success can be traced back to both the short development time of programs even for programmers without experience in parallel and distributed programs, and the fast and fault tolerant execution of many tasks.

However, there is also criticism passed on current progression towards MapReduce~\cite{MR_StepBackwards,MR_Comparison09,MR_Comparison10}.
This includes criticism on the applicability of MapReduce in all its simplicity to tasks where more evolved techniques have been examined already.
Hence, it is of high importance to gain an understanding when and when not the MapReduce model can lead to implementations that are efficient in practise.
In this spirit, MapReduce has been compared to PRAM~\cite{MR_Model} and BSP~\cite{MR_Nodari11} by presenting simulations in MapReduce. 
But theoretical foundations are still evolving.

Especially in high performance computing, the number of gigaflops provided by today's hardware is more than sufficient, and delivering data, i.e., memory access, usually forms the bottleneck of computation.
In clusters consisting of a large number of machines, the communication introduces another bottleneck of a similar kind.
Therefore, we believe that the (parallel) external memory model~\cite{av88,pem} provides an important context in which MapReduce should be examined.

In this paper, we provide further insights, contributing to the discussion on applicability of MapReduce, in that we shed light on the I/O-efficiency loss when expressing an algorithm in MapReduce.
On the other hand, our investigation bounds the complexity that can be ``hidden'' in the framework of MapReduce in comparison to the parallel external memory (PEM) model.
This serves for a direct lower bound on the number of rounds given a lower bound on the I/O-complexity in the PEM model.
The main technical contribution of this work is the consideration of the \emph{shuffle step} which is the single communication phase between processors / workers during a MapReduce round.
In this step, all information is redistributed among the workers.

\paragraph{MapReduce Framework}
The MapReduce framework can be understood as an interleaved model of parallel and serial computation.
It operates in rounds where within one round the user-defined serial functions are executed independently in parallel.
Each round consists of the consecutive execution of a map, shuffle and reduce step.
The input is a set of $\langle key, value \rangle$ pairs. 
Since each mapper and reducer is responsible for a certain (known) key, w.l.o.g. we can rename keys to be contiguous and starting with one.

A round of the MapReduce framework begins with the parallel execution of independent \emph{map} operations.
Each map operation is supplied with one $\langle key, value \rangle$ pair as input and generates a number of intermediate $\langle key, value \rangle$ pairs.
To allow for parallel execution, it is important that map operations are independent from each others and rely on a single input pair.
In the \emph{shuffle} step, the set of all intermediate pairs is redistributed s.t. lists of pairs with the same key are available for the reduce step.
The \emph{reduce} operation for key $k$ gets the list of intermediate pairs with key $k$ and generates a new (usually smaller) set of pairs. 

The original description, and current implementations realise this framework by first performing a \emph{split} function to distribute input data to workers.
Usually, multiple map and reduce tasks are assigned to a single worker.
During the map phase, intermediate pairs are already partitioned according to their keys into sets that will be reduced by the same worker.
The intermediate pairs still reside at the worker that performed the map operation and are then pulled by the reduce worker.
Sorting the intermediate pairs of one reduce worker by key finalises the shuffle phase.
Finally, the reduce operations are executed to complete the round.
A common extension of the framework is the introduction of a \emph{combiner} function that is similar in spirit to the reduce function.
However, a combine function is already applied during the map execution, as soon as enough intermediate pairs with the same key have been generated.

Typically, a MapReduce program involves several rounds where the output of one round's reduce functions serves as the input of the next round's map functions~\cite{MapReduce10}.
Although most examples are simple enough to be solved in one round, there are many tasks that involve several rounds such as computing page rank or prefix sums.
In this case a consideration of the shuffle step becomes most important,
especially when map and reduce are I/O-bounded by writing and reading intermediate keys. 
If map and reduce functions are hard to evaluate and large data sets are reduced in their size by the map function, it is important to find evolved techniques for the evaluation of these functions.
However, this shall not be the focus of our work.

One can see the shuffle step as the transposition of a (sparse) matrix:
Considering columns as origin and rows as destination, there is a non-zero entry $x_{ij}$ iff there is a pair $\langle i, x_{ij} \rangle$ emitted by the $j$th map operation (and hence will be sent to reducer $i$).
Data is first given partitioned by column, and the task of the shuffle step is to reorder non-zero entries row-wise.
Note that there is consensus in current implementations to use a partition operation during the map operation as described above.
This can be considered as a first part of the shuffle step.

\paragraph{Related work}
Feldman et al.~\cite{MR_Streaming} started a first theoretical comparison of MapReduce and streaming computation.
They address the class of symmetric functions (that are invariant under permutation of the input) and restrict communication and space for each worker to be polylogarithmic in the input size $N$ (but mention that results extend to other sublinear functions).
In~\cite{MR_Model}, Karloff et al. state a theoretical formulation of the MapReduce model where space restriction and the number of workers is limited by $\O{N^{1-\epsilon}}$.
Similarly, space restrictions limit the number of elements each worker can send or receive.
In contrast to other theoretical models, they allow several map and reduce tasks to be run on a single worker.
For this model, they present an efficient simulation for a subclass of EREW PRAM algorithms.
Goodrich et al.~\cite{MR_Nodari11} introduce the parameter $M$ to restrict the number of elements sent or received by a machine.
Similar to the external memory model where computation is performed in a fast cache of size $M$, this introduces another parameter additionally to the input size $N$.
Their MapReduce model compares to the BSP model with $M$-relation, i.e., a restricted message passing degree of $M$ per super-step.
The main difference, is that in all the MapReduce models information cannot reside in memory of a worker, but has to be resent to itself to be preserved for the next round.
A simulation of BSP and CRCW PRAM algorithms is presented based on this model.

The restriction of worker-to-worker communication allows for the number of rounds to be a meaningful performance measure.
As observed in~\cite{MR_Model} and~\cite{MR_Nodari11}, without restrictions on space / communication there is always a trivial non-parallel one-round algorithm where a single reducer performs a sequential algorithm.

On the experimental side, MapReduce has been applied to multi-processor / multi-core machines with shared memory~\cite{MR_MultiCore}.
They found several classes of problems that perform well in MapReduce even on a single machine.

As described above, the shuffle step can be considered as a matrix transposition.
Following the model of restricting communication, the corresponding matrix is restricted to have a certain maximum number of non-zero entries per column and row.
Previous work considered the multiplication of a sparse matrix with one vector~\cite{TCS10} and several vectors simultaneously~\cite{mfcs10} in the external memory model.
The I/O-complexity of transposing a dense matrix was settled in the seminal paper by Aggarval and Vitter~\cite{av88} that introduced the external memory model.
A parallel version of the external memory model was proposed by Arge et al.~\cite{pem}.

\paragraph{Our contribution}
We provide upper and lower bounds on the parallel I/O-complexity of the shuffle step.
In this, we can show that current implementations of the MapReduce model as a framework are almost optimal in the sense of worst-case asymptotic parallel I/O-complexity.
This further yields a simple method to consider the external memory performance of an algorithm expressed in MapReduce.
Since we consider the PEM model, we assume a shared memory.
However, our results can be applied to any layer of the memory hierarchy, and also extend to the communication layer, i.e. the network.
This is expressed by a comparison to the BSP$^*$ model.

Following the abstract description of MapReduce~\cite{MR_Nodari11,MR_Model}, the input of each map function is a single $\langle key, value\rangle$ pair.
The output of reduce instead can be any finite set of pairs.
In terms of I/O complexity, however, it is not important how many pairs are emitted, but rather the size of the input / output matters.

We analyse several different types of map and reduce functions.
For map, we first consider an arbitrary order of the emitted intermediate pairs.
This is most commonly phrased as the standard shuffle step provided by a framework.
Another case is that intermediate pairs are emitted ordered by their key.
Moreover, as a last case, we allow evaluations of a map function in parallel by multiple processors.
For reduce, we consider the standard implementation which guarantees that a single processor gets data for the reduce operations ordered by intermediate key.
Additionally, we consider another type of reduce which is assumed to be associative and parallelisable. 
This is comparable to the combiner function described before (cf. \cite{MapReduce04}).
In this case, the final result of reduce will be created by a binary tree-like reduction of the partial reduce results that were generated by processors holding intermediate results from the same key.
For the cases where we actually consider the evaluation of map and reduce functions, we assume that input / output of a single map / reduce function fits into internal memory.
We further assume in these cases that input and output read by a single processor does not exceed the number of intermediate pairs it accesses.
Otherwise, the complexity of the task can be dominated by reading the input, writing the output respectively, which leads to a different character that is strongly influenced by the implementation of map and reduce.
For the most general case of MapReduce, we simply assume that intermediate keys have already been generated by the map function, and have to be reordered to be provided as a list to the reduce workers.

For our lower bounds to be matching, we have to assume that the number of messages sent and received by a processor is restricted. 
More precisely, for $N_M$ being the number of map operations and $N_R$ the number of reducers, we require that each reducer receives at most $N_M^{1-\gamma}$ intermediate pairs, and each mapper emits at most $N_R^{1- \gamma}$ where $\gamma$ depends on the type of map operation.
However, for the first and the second types of map as described above, any $\gamma > 0$ is sufficient.

\paragraph{Outline}
In the next Section, we give a description of the PEM.
This will be followed by a comparison of the PEM and BSP$^*$ model.
In Section~\ref{sec:algos}, we present algorithms for the shuffle step for all considered map and reduce types.
Our lower bounds that match up to constant factors are given in Section~\ref{sec:lower_bounds}.



\section{The parallel external memory model}
\label{sec:pem}
The classical external memory model introduced in~\cite{av88} assumes a two-layer memory hierarchy.
It consists of an internal memory (cache) that can hold up to $M$ elements, and an external memory (disk) of infinite size, organised in blocks of $B$ elements.
Computations and rearrangement of elements can only be done with elements residing in internal memory.
With one I/O, a block can be moved between internal and external memory.
This models quite well the design of current hardware with a hierarchy of faster and smaller caches towards the CPU.
Disk accesses are usually cost-intensive and hence many contiguous elements are transferred at the same time.

As a parallel version of this model, the PEM was proposed by Arge et al.~\cite{pem} replacing the single CPU-cache by $P$ parallel caches and CPUs that operate on them (cf. Figure~\ref{fig:pem}).
External memory is treated as shared memory, and within one parallel I/O, each processor can perform an input or an output of its internal memory to disk.
Similar to the PRAM model, one has to define how overlapping access is handled.
In this paper, we assume concurrent read, exclusive write (CREW).
However, the results can be easily modified for CRCW or EREW.

%
\begin{figure}[!hbt]
	\begin{center}
	\begin{tikzpicture}[scale=.3]
		\draw (-18,0) rectangle (-11,1);
		\foreach \pos in {1,...,6} \draw (\pos-18,0)--(\pos-18,1);
		\node at (-15,1.35) [scale=.6, single arrow, fill=gray!30!white, rotate=90, single arrow head extend=4, single arrow head indent=.5] {};
		\node at (-14,1.65) [scale=.6, single arrow, fill=gray!30!white, rotate=270, single arrow head extend=4, single arrow head indent=.5] {};
		\draw[xshift=-14.5cm] (-.9,3.7) -- (-.9,1.8) (-.6,3.7) -- (-.6,1.8) (-.3,3.7) -- (-.3,1.8) (0,3.7) -- (0,1.8) (.3,3.7) -- (.3,1.8) (.6,3.7) -- (.6,1.8) (.9,3.7) -- (.9,1.8);
		\draw[fill=gray, thick, rounded corners] (-15.7,2) rectangle (-13.3,3.5);
		\draw (-14.5,2.75) node {\tiny \textbf{CPU}};
		\node at (-13.5,-1.2) [scale=1.4, single arrow, fill=gray!30!white,  rotate=90, single arrow head extend=4, single arrow head indent=.5] {};
		\node at (-15.5,-.8) [scale=1.4, single arrow, fill=gray!30!white, rotate=270, single arrow head extend=4, single arrow head indent=.5] {};
		
		\draw (-9,0) rectangle (-2,1);
		\foreach \pos in {1,...,6} \draw (\pos-9,0)--(\pos-9,1);
		\node at (-6,1.35) [scale=.6, single arrow, fill=gray!30!white, rotate=90, single arrow head extend=4, single arrow head indent=.5] {};
		\node at (-5,1.65) [scale=.6, single arrow, fill=gray!30!white, rotate=270, single arrow head extend=4, single arrow head indent=.5] {};
		\draw[xshift=-5.5cm] (-.9,3.7) -- (-.9,1.8) (-.6,3.7) -- (-.6,1.8) (-.3,3.7) -- (-.3,1.8) (0,3.7) -- (0,1.8) (.3,3.7) -- (.3,1.8) (.6,3.7) -- (.6,1.8) (.9,3.7) -- (.9,1.8);
		\draw[fill=gray, thick, rounded corners] (-6.7,2) rectangle (-4.3,3.5);
		\draw (-5.5,2.75) node {\tiny \textbf{CPU}};
		\node at (-4.5,-1.2) [scale=1.4, single arrow, fill=gray!30!white,  rotate=90, single arrow head extend=4, single arrow head indent=.5] {};
		\node at (-6.5,-.8) [scale=1.4, single arrow, fill=gray!30!white, rotate=270, single arrow head extend=4, single arrow head indent=.5] {};
		
		\draw (0,0) rectangle (7,1);
		\foreach \pos in {1,...,6} \draw (\pos,0)--(\pos,1);
		\node at (3,1.35) [scale=.6, single arrow, fill=gray!30!white, rotate=90, single arrow head extend=4, single arrow head indent=.5] {};
		\node at (4,1.65) [scale=.6, single arrow, fill=gray!30!white, rotate=270, single arrow head extend=4, single arrow head indent=.5] {};
		\draw[xshift=3.5cm] (-.9,3.7) -- (-.9,1.8) (-.6,3.7) -- (-.6,1.8) (-.3,3.7) -- (-.3,1.8) (0,3.7) -- (0,1.8) (.3,3.7) -- (.3,1.8) (.6,3.7) -- (.6,1.8) (.9,3.7) -- (.9,1.8);
		\draw[fill=gray, thick, rounded corners] (2.3,2) rectangle (4.7,3.5);
		\draw (3.5,2.75) node {\tiny \textbf{CPU}};
		\node at (4.5,-1.2) [scale=1.4, single arrow, fill=gray!30!white,  rotate=90, single arrow head extend=4, single arrow head indent=.5] {};
		\node at (2.5,-.8) [scale=1.4, single arrow, fill=gray!30!white, rotate=270, single arrow head extend=4, single arrow head indent=.5] {};
		
		\node at (9, .5) {$\dots$};
		
		\draw (11,0) rectangle (18,1);
		\foreach \pos in {1,...,6} \draw (\pos+11,0)--(\pos+11,1);
		\node at (14,1.35) [scale=.6, single arrow, fill=gray!30!white, rotate=90, single arrow head extend=4, single arrow head indent=.5] {};
		\node at (15,1.65) [scale=.6, single arrow, fill=gray!30!white, rotate=270, single arrow head extend=4, single arrow head indent=.5] {};
		\draw[xshift=14.5cm] (-.9,3.7) -- (-.9,1.8) (-.6,3.7) -- (-.6,1.8) (-.3,3.7) -- (-.3,1.8) (0,3.7) -- (0,1.8) (.3,3.7) -- (.3,1.8) (.6,3.7) -- (.6,1.8) (.9,3.7) -- (.9,1.8);
		\draw[fill=gray, thick, rounded corners] (13.3,2) rectangle (15.7,3.5);
		\draw (14.5,2.75) node {\tiny \textbf{CPU}};
		\node at (15.5,-1.2) [scale=1.4, single arrow, fill=gray!30!white,  rotate=90, single arrow head extend=4, single arrow head indent=.5] {};
		\node at (13.5,-.8) [scale=1.4, single arrow, fill=gray!30!white, rotate=270, single arrow head extend=4, single arrow head indent=.5] {};

		\draw[thick, lightgray!1] (-20,-3) -- (20,-3);
		\draw (-20,-2) -- (20,-2);
		\draw (-20,-3) -- (20,-3);
		\draw (-21,-2.5) node {$\cdots$};
		\draw (21,-2.5) node {$\cdots$};
		\foreach \pos in {1,...,39} \draw (\pos-20,-2)--(\pos-20,-3);
		\foreach \pos in {1,...,8} \draw[very thick] (\pos*5-22,-2)--(\pos*5-22,-3);
		\draw[thick, draw=black] (-20,-2) -- (20,-2);
		\draw[thick, draw=black] (-20,-3) -- (20,-3);

	\end{tikzpicture}
	\caption{The parallel external memory model (PEM).}
	\label{fig:pem}
	\end{center}
\end{figure}
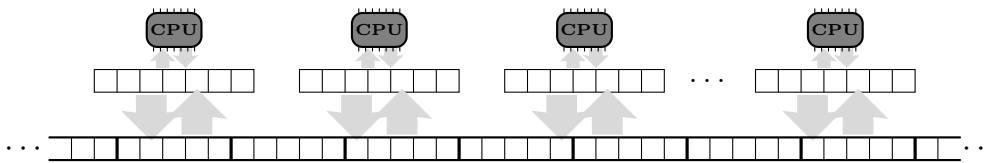

\section{A comparison to BSP$^*$}
For a BSP model comparable with PEM, we assume that computational costs are dominated by communication.
One of the parameters in the BSP model is the $h$-relation, the number of messages each processor is allowed to send and to receive.
Additionally, a latency / startup cost $L$ per super-step is assumed.
The total cost of an algorithm with $T$ super-steps is $T\cdot(h + L)$.
All this conforms with the definition of BSP in~\cite{bsp}, when normalising to network throughput $g=1$.
However, as a significant change compared to~\cite{bsp} one can define the latency $L$ in relation to the number of connections each processor has to establish.
This is justified for hand shake protocols but also for the encapsulation process of messages performed by the network layers in the OSI model, i.e., todays network protocols.
Hence, an incentive is given to send a number of elements of magnitude comparable to the connection-latency to the same processor.
Another way to express this is the BSP$^*$ model~\cite{bsp_star} that encourages block-wise communication by defining a cost model $g h \ceil{s/B} + L$ per super-step for maximum message length $s$.

The version that is best comparable with PEM is the BSP$^*$ with 1-relation (i.e., $h=1$).
Assuming $g > L$, we can restrict ourselves to $s \leq B$.
Otherwise a message has to be divided into multiple super-steps which only changes costs by a constant factor.
Any such 1-BSP$^*$ algorithm with $\ell$ super-steps on input size $N$ can be simulated in the EREW PEM with $2 \ell$ parallel I/O, if input and output are equally distributed over internal memories and $M$ is sufficiently large.
In this, one parallel-output of the blocks that are to be sent is done per super-step.
In a following parallel input, these blocks are input by their destined processor.
Hence, all our lower bounds hold for the 1-BSP$^*$ where $M$ can be set to any arbitrarily large value (e.g. $N$) such that $N/P \leq M$.
Note that the simulation can be non-uniform which, however, does not change the statement.

Similarly, a 1-BSP$^*$ algorithm can be derived from an EREW PEM algorithm for $M = N/P$.
For each output that is made by the PEM algorithm, the corresponding block is sent to the processor that will read the block nearest in the future.
In general, this implies that multiple blocks can be sent to the same processor violating the 1-relation.
However, our algorithms are easily transformed to avoid this problem. 
Furthermore, assignment of input and output is not necessary in the BSP model, and the applied parallel sorting algorithm is even derived from a BSP algorithm (see~\cite{pem}).
We omit a detailed description here since it is not the focus of our work.


\section{Upper bounds for the shuffle step}
\label{sec:algos}
We start with the description of some basic building blocks that are performed multiple times within our algorithms.
To form a block from elements that are spread over several internal memories, processors can communicate elements in a tree-like fashion to form the complete block in $\O{\log \mins{P,B}}$ I/Os.
This is referred to as \emph{gathering}.
If a block is created by computations involving elements from several processors (e.g., summing multiple blocks) still $\O{\log P}$ I/Os are sufficient.
Similarly, a block can be spread to multiple processors in $\O{\log P}$ I/Os (\emph{scattering}).
\if@fullpaper
If $n$ blocks for each processor have to be scattered / gathered independently, this can be serialised such that $\O{n + \log P}$ I/Os are sufficient.
Note that for all gather and scatter tasks, the communication structure has to be known to each participant.
This is for example the case when participating processors constitute an ordered set which is known to all.
\fi
Additionally, we require the computation of prefix sums.
This task has been extensively studied in parallel models.
For the PEM model see~\cite{pem} for a description.

\if@fullpaper
\paragraph{Range-bounded load-balancing}
\label{sec:RVdist}
For load balancing reasons, the following task will appear several times. 
Given $n$ tuples $(i, x)$ that are located in contiguous external memory, ordered by some non-unique key $i \in \{1,\dots,m\}$ for $m \leq n$.
Assign the $n$ tuples to $P$ processors such that each processor gets $\O{n/P}$ tuples assigned to it, but keys are within a range of size $\O{m/P}$.

The task can be solved by dividing the $P$ processors into $\lceil P/2 \rceil$ \emph{volume processors} 
 and $\lfloor P/2 \rfloor$ \emph{range processors}. 
First, data is assigned to the volume processors by volume.
In this, each volume processor gets a consecutive piece of $\lceil 2n/P \rceil$ tuples assigned to it. 
For communication purpose, we assume that for each processor there is an exclusive block reserved for messages in external memory denoted as \emph{inbox}.

For the next step, think of the ordered set of keys $(1,\dots,m)$ being partitioned into $P/2$ ranges of at most $\lceil 2m/P \rceil$ continuous keys each (further referred to as \emph{key range}).
To cope with the problem of having too many different keys assigned to the same volume processor, we reserve the $i$th range processor to become responsible for tuples within the $(i+1)$th key range. 
Now, each volume processor scans its assigned area and keeps track of the position in external memory where a new key range begins.
This requires only $\O{\frac{n}{PB}}$ I/Os for scanning the assigned tuples while one block in memory is reserved to buffer and output the starting positions of a new range.

Afterwards, each volume processors with more than $\ceil{2m/P}$ keys assigned to it, created a list of memory positions where a new key range begins.
This list is than scattered block-wise to the range processors reserved for the corresponding key ranges. 
Although we assume CREW, it is necessary to distribute this information because a program running on a range processor is not aware of the memory position to find the information for it depends on which volume processor got assigned the key range.
The distribution is possible in $\O{\log \mins{B,P}}$ parallel I/Os where each block of a list is written into the inbox of the first range processor concerned by this block, and from there spread in an binary manner to the other range processors (cf. Figure~\ref{fig:scatter}).
In this, each processor can be informed about both the beginning and the end of its assigned area.
Note that it can never happen that a range processor receives information from multiple volume processors because for a key range divided to multiple volume processors only the first can dispose the tuples.
\fi

\begin{figure}[ht]
	\begin{center}
		\begin{tikzpicture}[scale=.6]
			\draw[fill=lightgray] (-.2,0) rectangle (1.2,1);
			\node at (.5,.5) {\scriptsize $P_i$};

			\draw[fill=lightgray] (1.8,0) rectangle (3.2,1);
			\node at (2.5,.5) {\scriptsize $P_{i+1}$};

			\draw[fill=lightgray] (3.8,0) rectangle (5.2,1);
			\node at (4.5,.5) {\scriptsize $P_{i+2}$};
			
			\node at (6,.5) {\scriptsize $\dots$};

			\draw[fill=lightgray] (6.8,0) rectangle (8.2,1);
			\node at (7.5,.5) {\scriptsize $P_{i+\frac{B}{2}}$};

			\node at (9,.5) {\scriptsize $\dots$};

			\draw[fill=lightgray] (9.8,0) rectangle (11.2,1);
			\node at (10.5,.5) {\scriptsize $P_{i+B}$};
			
			\draw[fill=gray] (-1.6,2.1) rectangle (.5,2.4);
			\draw (.2,2.1) -- (.2,2.4)
				  (-.1,2.1) -- (-.1,2.4)
				  (-.4,2.1) -- (-.4,2.4)
				  (-.7,2.1) -- (-.7,2.4)
				  (-1,2.1) -- (-1,2.4)
				  (-1.3,2.1) -- (-1.3,2.4);
			
			\draw[->] (-.4,2) -- (.4,1.1);
			 \draw[dashed, ->] (.6,1.1) .. controls (1,2.4) and (2,2.4) .. (2.4,1.1);
			  \draw[dotted, ->] (2.55,1.1) .. controls (3,2) and (4,2) .. (4.5,1.1);
			  \draw[dotted, ->] (2.55,1.1) .. controls (3,2) and (5.5,2) .. (6,1.1);
			 \draw[dashed, ->] (.6,1.1) .. controls (.2,2.4) and (7.5,2.4) .. (7.4,1.1);
			  \draw[dotted, ->] (7.55,1.1) .. controls (8,2) and (8.3,2) .. (8.6,1.1);
			  \draw[dotted, ->] (7.55,1.1) .. controls (8,2) and (9,2) .. (9.5,1.1);
		\end{tikzpicture}
	\end{center}	
	\caption{Scattering the list of key ranges.}
	\label{fig:scatter}
\end{figure}
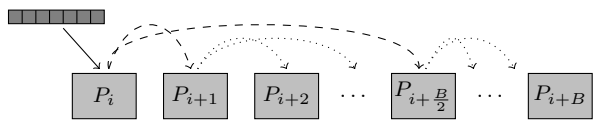

\if@fullpaper
\paragraph{Contraction}
Given an input sequence of $n$ elements written in $m$ blocks on disk, where some of the memory cells are empty, it is possible to contract the sequence in $\O{m/P + \log P}$ I/Os such that empty cells are removed and the sequence is written in the same ordering as before, but with elements stored contiguously.
To this end, the input blocks are assigned equally to the $P$ processors.
Each processor gets a contiguous peace of up to $\ceil{m/P}$ blocks assigned to it, and processors are assigned in ascending order.
Within one scan, each processor can determine the number of non-empty cells contained in its assigned area.
With this information, and the given ordering of processors, using prefix sum computation, in $\O{\log P}$ I/Os, it is known to each processor where its elements shall be placed in a contiguous output.
For each block of the contiguous output, the processor with the first, and the processor with the last element assigned to it are in knowledge of the fact, given the prefix sum results.
Since processors are assigned in ascending order to the pieces, a gather operation can be used to create blocks with elements assigned to several processors.
Note that each processor only needs to participate in at most two gather operations.
The gather process can be achieved by an output of each first and last processor of a block, in that both write their indices into a designated table.
Then, each processor can read this information and send its elements to the inbox of the responsible processor in the gather process.
\fi

\subsection*{Algorithms}
For a clearer understanding of the shuffle step, we use the analogy of a sparse matrix.
Let $N_M$ be the number of distinct input keys, and $N_R$ be the number of distinct intermediate keys (i.e., independent reduce runs).
Each pair $\langle i , x_{ij} \rangle$ emitted by map operation $j$ can be considered a triple $(i,j,x_{ij})$. 
Using this notation, one can think of a sparse $N_R \times N_M$ matrix with non-zero entries $x_{ij}$.
This matrix is given in some layout determined by the map function and has to be either reordered into a row-wise ordering, or a layout where rows can be reduced easily.
In the following, we consider input keys as column indices and intermediate keys as row indices.
The total number of intermediate pairs / non-zero elements is denoted $H$.
Additionally, we have $w$, the number of elements emitted by a reduce function, and $v$, the size of the input to a map function, $v,w \leq \mins{M-B, \ceil{H/P}}$ as argued in the introduction.
An overview of the algorithmic complexities is given in Table~\ref{tab:overview}.
For space restrictions the term $\log P$ is omitted in Table~\ref{tab:overview}.
We use $\lgb_b x := \maxs{\log_b x, 1}$.
The complexities given in Table~\ref{tab:overview} only differ from the descriptions in that we distinguish the special case $R = 1$, and make use of the observation $\O{\log_d (x/d)} = \O{\log_d x}$.
For all our algorithms we assume $H/P \geq B$, i.e., there are less processors than blocks in the input such that each processor can get a complete block assigned to it.

\begin{table}[h]
	\center
	{\scriptsize
	\renewcommand*\arraystretch{1.9}
	\begin{tabular}{|l||c|c|}
		\hline
		 & Non-parallel reduce & Parallel reduce \\		
		\hline
		\hline
		Unordered map & $\O{\frac{H}{PB} \lgb_d N_R}$ & $\O{\frac{H}{PB} \lgb_d \frac{N_R w}{B}}$  \\
		\hline
		Sorted map & $\O{\frac{H}{PB} \lgb_d \mins{\frac{N_M N_R B}{H}, N_R, N_M}}$ & $\O{\frac{H}{PB} \lgb_d \mins{\frac{N_M N_R w}{H}, \frac{N_R w}{B}}}$ \\ 
		\hline
		Parallel map & $\O{\frac{H}{PB} \lgb_d \mins{\frac{N_M N_R v}{H}, \frac{N_M v}{B}}}$ & $\O{\frac{H}{PB} \lgb_d \frac{N_M N_R vw}{BH}}$ \\
		\hline	
		\hline
		Direct shuffling & \multicolumn{2}{c|}{$\O{H/P}$ \ \ \scriptsize (non-uniform)}  \\
		Complete merge & \multicolumn{2}{c|}{$\O{\frac{H}{PB} \lgb_d \frac{H}{B}}$} \\
		\hline			
	\end{tabular}}
	\caption{Overview of the algorithmic complexities with $d = \mins{M/B, H/(PB)}$.}
	\label{tab:overview}
\end{table}

\subsubsection{Direct shuffling}
Obviously, the shuffle step can be completed by accessing each element once and writing it to its destination.
For a non-uniform algorithm, that is, with knowledge of the non-zero positions, $\O{H/P}$ parallel I/Os are sufficient.
To this end, the output can be partitioned into $H/P$ consecutive parts.
Since we assume $H/P \geq B$, collisions when writing can be avoided.
In contrast, because we consider CREW reading the elements in order to write them to their destination is possible concurrently.
We restrict ourselves in this case to a non-uniform algorithm to match the lower bounds in Section~\ref{sec:lower_bounds}.
This shows that our lower bounds are asymptotically tight.
Such a direct shuffle approach can be optimal.
However, for other cases that are closer to real world parameter settings a more evolved approach is presented in the following.

\subsubsection{Map-dependent shuffle part}
In this part, we describe for different types of map functions how to prepare intermediate pairs to be reduced in a next step.
To this end, during the first step $R$ meta-runs of non-zero entries from ranges of different columns will be formed. 
Afterwards, these meta-runs are further processed to obtain the final result.
The meta-runs shall be internally ordered row-wise (aka \emph{row major layout}).
If intermediate pairs have to be written in sorted order before the reduce operation can be applied, we set $R = \ceil{\frac{H}{N_R B}}$.
Otherwise, if the reduce function is associative, it will suffice to set $R = \ceil{\frac{H}{N_R \maxs{w,B}}}$. 

\paragraph{Non-parallel map, unordered intermediate pairs}
We first consider the most general (standard) case of MapReduce where we only assume that intermediate pairs from different map execution are written in external memory one after another.
The elements are ordered by column but within a column no ordering is given.
We refer to this as \emph{mixed column layout}.
This given ordering can only be exploited algorithmically for non-parallel reduce functions where a row major layout has to be constructed.

We apply the parallel merge sort by Arge et al. \cite{pem} to sort elements by row index.
This algorithm divides data evenly upon the $P$ processors and starts with an optimal (single processor) external memory algorithm such as the merge sort in~\cite{av88} to create $P$ presorted runs of approximately even size.
Then, these runs are merged in a parallel way with a merging degree $\ceil{d}$ for $d = \maxs{2,\mins{H/(PB), \sqrt{H/P}, M/B}}$. 
In contrast to~\cite{pem}, we slightly changed this merging degree in that we added the term $H/(PB)$ to the minimum.
This widens the parameter range given in~\cite{pem} without changing the complexity within the original range and guarantees matching lower bounds. 
Note that $\log H/(PB) \leq 2 \log \sqrt{H/P}$ such that for asymptotic considerations $d = \maxs{2, H/(PB), M/B}$ is sufficient.
Instead of a full merge sort, we stop the merging process when the number of runs is less than $R$.
Note that if $P \leq R$, we actually skip the parallel merging and perform only a local merge sort on each processor.
In either case, we get a parallel I/O-complexity of $\O{\frac{H}{PB} \lgb_d \frac{H}{BR} + \log P}$.

\paragraph{Non-parallel map, sorted intermediate pairs}
Here, we assume that within a column, elements are additionally ordered by row index, i.e. intermediate pairs are emitted sorted by their key.
This corresponds to the \emph{column major layout}.
In the following, we assume $H/N_R \geq B$. 
Otherwise, the previous algorithm is applied, or simply columns are merged together as described in a later paragraph. 

Since columns are ordered internally, each column can serve as a presorted run.
\if@fullpaper
Thus, we assign elements using the range-bounded load-balancing method, with column indices serving as keys.
With this assignment we run the parallel merge sort, but with the following modifications.
Again, we stop the merging processes as soon as less than $R$ runs remain.
Instead of the general external memory sorting algorithm that is used locally on each processor, we use the $M/B$-way merge sort to merge all the up to $N_M/P$ columns that are assigned to one processor.
For $P \leq R$ this local merge sort finishes the task by reducing the $N_M$ columns into $R$ runs.
If $P > R$ the local merging creates $P$ runs out of the $N_M$ columns.
Then, in the parallel merge phase, the $P$ runs are reduced to $R$ runs.
The total I/O-complexity of this algorithm is $\O{\frac{H}{PB} \lgb_d \frac{N_M}{R} + \log P}$.
\fi

\paragraph{Parallel map, sorted intermediate pairs}
In the following, we describe the case with the best possible I/O-complexity for the shuffle step when no further restrictions on the distribution of intermediate keys are made.
This is the only case where we actually consider the execution of the map function itself.
Note that in terms of I/O-complexity, an algorithm can emit intermediate pairs from a predefined key range only.
This is possible since intermediate pairs are generated in internal memory, and can be removed immediately without inducing an I/O, while pairs within the range of interest are kept.
In a model with considerations of the computational cost, it would be more appropriate to consider a map function which can be parallelised to emit pairs in a predefined intermediate key range.

We first describe the layout of intermediate pairs in external memory that shall be produced.
Let $m = \mins{M-B, \ceil{H/P}}$.
This time, intermediate pairs are not simply ordered by column primarily as before, but data is split into ranges of $m/v$ columns and within each such meta-column elements are ordered row-wise.
When creating this layout, each processor can keep its memory filled with the $m$ input elements required for each meta-column while writing the intermediate results.
%
\if@fullpaper
This layout can be obtained by first determining the volume of each meta-column in parallel by requesting intermediate results without writing them to disk.
Using parallel prefix sum computation, it can be determined which processor will be assigned to which volume part of a meta-column in order to realise an equal load balancing.
This step is possible in $\O{\log P}$ I/Os.
After the assignment of processors to meta-columns, each processor reads and keeps input pairs for a whole meta-column in internal memory.
Using the map function, intermediate pairs are requested and extracted, and then written to external memory in row-wise order (within a meta-column).
If a processor is assigned to multiple meta-columns, it processes each meta-column one after another.
Since $H/P \geq B$, multiple processors never have to write to the same block at the same time.

Because we already formed row-wise sorted meta-columns of $m/v$ columns, the number of merge iterations to generate $R$ row-wise sorted meta-runs is reduced.
If $N_M/m \leq R$, nothing needs to be done because the number of meta-columns is already less than the desired number of meta-runs.
Otherwise, we use the parallel merge sort.
Similar to the description for sorted intermediate pairs, the elements are assigned to processors in a balanced way, range-bounded now by meta-column index.
As a single processor sorting algorithm, the meta-columns are merged by merge-sort (if less than $P$ meta-columns exist).
The local or the preceding parallel merging is again stopped when at most $R$ runs remain.
\fi
Reducing the $N_Mv/m$ meta-columns into $R$ meta-runs induces an I/O-complexity of 
 $\O{\frac{H}{PB} \lgb_{M/B} \frac{N_Mv}{\mins{M, H/P} R} + \log P}$.

\subsubsection{Reduce-dependent shuffle part}
\emph{Non-parallel reduce function} For the general case of non-parallel reduce functions, intermediate keys of the same key are to be provided consecutively to the reduce worker, i.e., a row major layout has to be created.
This can be obtained in a direct manner as follows.
We describe the current layout in tiles, where one tile consists of the elements in one row within a meta-column.
The macroscopic ordering of these tiles is currently a column major layout.
To obtain the desired layout, tiles only need to be rearranged into a row major layout.

Observe that there are $\O{\frac{H}{N_R B}}$ meta-runs with $N_R$ rows each such that there are at most $\O{\frac{H}{B}}$ non-empty tiles.
\if@fullpaper
Each tile can consist of at most two blocks that contain elements from another tile (and need to be accessed separately).
However, these are still $\O{\frac{H}{PB}}$ parallel I/Os to access these blocks.
The remaining $\O{\frac{H}{B}}$ blocks that belong entirely to a tile contribute another $\O{\frac{H}{PB}}$ I/Os.

To rearrange tiles, we assign elements to processors balanced by volume and range-bounded by the tile index (ordered by meta-runs first, and by row within a meta-run).
In order to write the output in parallel, the destined positions of each of its assigned elements has to be known to the processors.
To this end, each processor scans its assigned elements, and for each beginning of a new tile, the memory position is output into a table $S$ (consisting of $\O{H/B}$ entries, one blocks each).
Afterwards, using this table, the size of each tile is determined and written to a new table $D$.
With table $D$ a prefix sum computation is started in row major layout such that $D$ now contains the relative output destination of each row within each meta-run.
In a CRCW model, with the same assignment of elements as before, tiles can now be written to their destination to form the output using table $D$.

For CREW, when first creating $D$, one can ceil the tile sizes to full blocks.
The resulting layout will obviously contain blocks that are not entirely filled, but contain empty memory cells.
However, using the contraction described above, one can extract these empty cells.

\fi
The whole step to finalise the shuffle step has I/O-complexity $\O{\frac{H}{PB} + \log P}$.

\paragraph{Parallel (associative) reduce function}
Assuming a parallelisable reduce, each processor shall perform multiple reduce functions simultaneously on a subset of elements with intermediate key in a certain range.
In a final step, the results of these partial reduce executions are then collected and reduced to the final result.

To this end, the range of intermediate keys is partitioned into $\ceil{N_RPw/H}$ ranges of up to $\ceil{H/(Pw)}$ keys.
\if@fullpaper
Using the range-bounded load-balancing algorithm, elements (still ordered in meta-runs) are assigned to processors such that each processor gets elements from at most two pieces of row indices.
This can be achieved by using the tuple $(\text{\it meta-run index}, \text{\it row index})$ as key.
\fi
If a processor got assigned elements that belong to the same reduce function, elements can be reduced immediately by the processor.
Afterwards, for each key range, elements can be gathered to form the final result of the reduce function.
This is possible with $\O{\frac{H}{PB} + \log P}$ I/Os.

\subsubsection{Complete sorting / merging}
For some choices of parameters, especially for small instances, it can be optimal to simply apply a sorting algorithm to shuffle elements row-wise.
Using the parallel merge sort, this has I/O-complexity $\O{\frac{H}{PB} \lgb_d \frac{H}{B} + \log P}$.
Furthermore, if the matrix is given in column major layout, the $N_M$ already sorted columns can simply be merged to order elements row-wise.
This results in an I/O-complexity of $\O{\frac{H}{PB} \lgb_d N_M + \log P}$.

\section{Lower bounds for the shuffle step}
\label{sec:lower_bounds}
A simple task in MapReduce is creating the product of a sparse matrix~$A$ with a vector.
Assuming that the matrix entries are implicitly given by the map function, the task can be accomplished within one round. 
In this, map function $j$ is supplied with input vector element $x_j$ and emits $\langle i, x_j a_{ij} \rangle$.
The reduce function simply sums up values of the same intermediate key.
Hence, a lower bound for matrix vector multiplication immediately implies a lower bound for the shuffle step.
Since reduce can be an arbitrary function, we restrict ourselves to matrix multiplication in a semiring, where the existence of inverse elements is not guaranteed.


A lower bound for sparse $N \times N$ matrices in the I/O-model was presented in \cite{TCS10} and extended to non-square situations in \cite{mfcs10}.
These bounds are based on a counting argument, comparing the number of possible programs for the task with $\ell$ I/Os to the number of distinct matrices.
To this end, the maximal number of different configurations (content) of external and internal memory after one I/O is examined.
In this section, we explain the main differences to these proofs, the complete proofs of our results can be found in Appendix~\ref{app:lower_bounds}.


Since we have multiple processors and assume a CREW environment, we have to dissociate from the perspective of moving blocks.
Instead, we assume that a block in external memory which does not belong to the final output disintegrates magically immediately after it is read for the last time in this form.

For a parallel I/O, the number of preceding configurations is simply raised to the power of $P$.
However, in contrast to the single processor EM model, we have to consider the case $H/P \leq M$, i.e., not all processors can have their internal memory filled entirely.
Instead, we consider the current number of elements $X_{i,l}$ of processor $i$ before the $l$th parallel I/O.
The number of distinct configuration of a program after $\ell$ I/Os is then bounded by
\[ \prod_{l=1}^\ell 3^P \prod_{i=1}^P {X_{i,l} + B \choose B} 2^B 4P\ell \text{\,.}\]

Furthermore, we consider multiple input and multiple output vectors which leads to a combined matrix vector product.
In this, any intermediate pair -- in classical matrix vector multiplication an elementary product of a vector element with a non-zero entry -- can now be the result of a linear combination of the $v$ elements in the corresponding dimension of the input vectors, and any output element can be a linear combination of intermediate pairs with corresponding intermediate key.
We use a simplified version and consider the variant where each intermediate pair is simply a copy of one of the $v$ input elements, and it is required for the computation of precisely one output element.
Hence, for each of the $H$ non-zero entries in our matrix, there is not only a choice of its position but also the choice from which of the $v$ input elements it stems from, and which of the $w$ output elements will be its destination.
This results in a total number of $\binom{N_M N_R}{H} v^H w^H$ different tasks for fixed parameters $N_M$, $N_R$, $H$, $v$ and $w$.

Applying these modifications yields the following results.

\begin{theorem}
	\label{thm:SpMxV_lb}
	Given parameters $B$, $M \geq 3B$ and $P \leq \frac{H}{B}$.
	Creating the combined matrix vector product for a sparse $N_R \times N_M$ matrix with $H$ non-zero entries for $H/N_R \leq N_M^{1-\epsilon}$ and $H/N_M \leq N_R^{1-\epsilon}$ for $\epsilon > 0$ from $v \leq H/N_M$ input vectors to $w \leq H/N_R$ output vectors has (parallel) I/O-complexity
	\begin{itemize}
		\item $ \W{\mins{\frac{H}{P}, \frac{H}{PB} \log_d \frac{N_R w}{B}}} $ if the matrix is in mixed column layout
		\item $ \W{\mins{\frac{H}{P}, \frac{H}{PB} \log_d \mins{\frac{N_M N_R w}{H}, \frac{N_Rw}{B}}}} $ if given in column major layout
		\item and $ \W{\mins{\frac{H}{P}, \frac{H}{PB} \log_d \frac{N_M N_R vw}{H \mins{M, H/P}}}} $ for the best-case layout with $H/N_R \leq \sqrt[6]{N_M}$ and $H/N_M \leq \sqrt[6]{N_R}$ 
	\end{itemize}
	where $d = \mins{M/B, H/(PB)}$.
\end{theorem}
These lower bounds already match the algorithmic complexities for parallel reduce in Section~\ref{sec:algos}.
Moreover, a lower bound for creating a matrix in row major layout from $v$ vectors can be obtained in a very similar way (cf. parallel map \& non-parallel reduce).
This task corresponds to a time-inverse variant of creating the matrix vector product with a matrix in column major layout.

\begin{lemma}
	\label{lem:SpMxCreate}
	Given parameters $B$, $M \geq 3B$ and $P \leq \frac{H}{B}$.
	Creating a sparse $N_R \times N_M$ matrix with $H$ non-zero entries in row major layout from $v$ vectors $x^{(1)}, \dots, x^{(v)}$ such that for all non-zero entries holds $a_{ij} = x_j^{(k)}$ for some $k$  has (parallel) I/O-complexity
	\[ \W{\mins{\frac{H}{P}, \frac{H}{PB} \log_d \mins{\frac{N_M N_R v}{H}, \frac{N_M v}{B}}}} \]
	for $H/N_R \leq N_M^{1-\epsilon}$ and $H/N_M \leq N_R^{1-\epsilon}$, $\epsilon > 0$.
\end{lemma}

Theorem~\ref{thm:SpMxV_lb} and Lemma~\ref{lem:SpMxCreate} both hold not only in the worst-case, but for a fraction of the possible sparse matrices exponentially close to one.
Hence, for distributions over the matrix conformations (position of the non-zero entries), even if not uniform but somehow skewed, the lower bounds still hold on average if a constant fraction of the space of matrix conformations has constant probability.
Similar, the bounds hold on average for distributions where a constant fraction of the non-zero entries is drawn with constant probability from a constant fraction of the possible position.

\subsubsection{Transposing bound}
Another method to obtain a lower bound for the I/O-complexity is presented in~\cite{av88}.
They use a potential function to lower bound the complexity of dense matrix transposition.
This bound can also be extended to sparse matrix transposition in the PEM model for matrices given in column major layout.
Combining the following bound with Theorem~\ref{thm:SpMxV_lb} matches the algorithmic complexities given in Section~\ref{sec:algos} for non-parallel reduce.

\begin{theorem}
The transposition of a sparse $N_R \times N_M$ matrix with $H$ non-zero entries has worst-case parallel I/O-complexity
\[ \W{\frac{H}{PB} \log_d\mins{B, N_M, N_R, \frac{H}{B}}} \text{\,.}\]
\end{theorem} 

\subsubsection{A bound for scatter / gather}
To cover all the algorithmic complexities, it remains to justify the scatter and gather tasks that are required for the exclusive write policy.
A lower bound for sorting related problems can be found in~\cite{pem}.
They show a lower bound of $\W{\log N/B} = \W{\log P}$ on the number of I/Os.


\section{Conclusion}
We determined the parallel worst-case I/O-complexity of the shuffle step for most meaningful parameter settings.
All our upper and lower bounds for the considered variants of map and reduce functions match up to constant factors.
Although worst-case complexities are considered, most of the lower bounds hold with probability exponentially close to one over uniformly drawn shuffle tasks.
We considered several types of map and reduce operations, depending on the ordering in which intermediate pairs are emitted and the ability to parallelise the map and reduce operations.
All our results hold especially for the case where internal memory of the processors is never exceeded but (blocked) communication is required.
This shows that the parallel external memory model reveals a different character than the external memory model in that communication can be described in the model even for the case the input fits into internal memories.

Our results show that for parameters that are comparable to real world settings, sorting in parallel is optimal for the shuffle step.
This is met by current implementations of the MapReduce framework where the shuffle step consists of several sorting steps, instead of directly sending each element to its destination.
In practise one can observe that a merge sort usually does not perform well, but rather a distribution sort does.
The partition step and the network communication in current implementations to realise the shuffle step can be seen as iterations of a distribution sort.
Still, our bounds suggest a slightly better performance when in knowledge of the block size.
If block and memory size are unknown to the algorithm, which corresponds to the so called cache-oblivious model, it is known that already permuting ($N_M=N_R=H$) cannot be performed optimally.
Sorting instead can be achieved optimally, but only if $M\geq B^2$ \cite{BF03}.
However, when assuming that the naïve algorithm with $\O{\frac{H}{P}}$ I/Os is not optimal, and $M \geq B^2$, all the considered variants have the same complexity and reduce to sorting all intermediate pairs in parallel.

\bibliographystyle{is-abbrv}
\bibliography{../Literature/sparse.bib}

\newpage

\begin{appendix}
	
	\section{Lower bounds}
	\label{app:lower_bounds}
	A lower bound for sparse $N \times N$ matrices in the I/O-model was presented in \cite{TCS10} and extended to non-square situations in \cite{mfcs10}.
	We follow closely their description but have to restate the proof for the PEM.
	Additionally, we consider a task where multiple input and output elements are associated with each non-zero entry.
	More specifically, we have $v$ input and $w$ output vectors.
	The task is to multiply each non-zero entry $a_{ij}$ with only one of the $v$ vector elements $x_j^{(1)}, \dots, x_j^{(v)}$.
	The assignment which vector is used for which non-zero entry is part of the input.
	Furthermore, each non-zero entry is associated with one of the $w$ output vectors.
	The $i$th coordinate of the $l$th output vector is the sum of all elementary products $a_{ij} x_j^{(k)}$ that were associated with vector $l$.
	Hence, each non-zero entry has, apart from its position in the matrix and its value, two further variables assigned to it, defining origin of the input vector element and destination of the elementary product.
	We refer to this task as the combined matrix vector product of a given matrix, a set of $v$ input vectors, the number of output vectors $w$ and a given assignment of non-zero entries to input / output vectors.
	
	Like in~\cite{TCS10}, the configuration at time $t$ refers to the concatenation of non-empty memory cells of external and internal memory between the $t$th I/O and the $t+1$th I/O.
	We will need the following technical lemma for the proofs of the lower bounds.
	\begin{lemma}
		\label{lem:bounds}
		Assume $\log 3 H \geq \frac72 B \log \mins{\frac{M}{B}, \frac{2H}{PB}}$, $H \geq \maxs{N_1,N_2} \geq 2$, $H/N_2 \leq N_1^{1 - \epsilon}$, then
		\begin{itemize}
			\item[$(i)$] $H \leq N_2^{1/\epsilon}$
			\item[$(ii)$] $N_2 \geq 2^8$ implies $B \leq \frac{1}{e \epsilon} N_2^{3/8}$.
			\item[$(iii)$] $N_2 \geq 2^8$ and $H/N_2 \leq N_1^{1/6}$ implies $\mins{\frac{M}{B}, \frac{2H}{BP}} \leq N_2^\frac{3}{7B}$.
		\end{itemize}
	\end{lemma}
	\begin{proof}
		Combining $H/N_2 \leq N_1^{1 - \epsilon}$ with $H \geq N_1$ yields $N_1 \leq N_2 N_1^{1-\epsilon}$, i.e., $N_2 \leq N_1^{1/\epsilon}$.
		Substituting $N_1$ in $H \leq N_2 N_1^{1-\epsilon}$ results in $(i)$.

		For $(ii)$, we have $B \leq \frac27 \log 3 H \leq \frac1e \log H$ since $H \geq N_2 \geq 2^8$ such that we can use $3 \cdot 2^8 < 2^{10}$ and note that $\frac54 \cdot \frac27 \leq \frac1e$. Using $(i)$, we get $B \leq \frac1e \log N_2^{1/\epsilon} \leq \frac{1}{e \epsilon} \log N_2$.
		Finally, we simply use the additional observation $\log x \leq x^{3/8}$ for $x \geq 2^8$.

		The last results is obtained by rewriting the main assumption as $3 H \geq \mins{\frac{M}{B}, \frac{2H}{BP}}^{7B/2}$.
		Again, using $3H \leq H^{5/4}$ for $H \geq N_2 \geq 2^8$, we get $H \geq \mins{\frac{M}{B}, \frac{2H}{BP}}^{14B/5}$. 
		$H$ in term is bounded from above by $N_2^{6/5}$ such that we have $\mins{\frac{M}{B}, \frac{H}{BP}} \leq (N_2^{6/5})^{5/(14B)} \leq N_2^{3/(7 B)}$.
	\end{proof}

	\subsection*{Best-case to row major layout with multiple input pairs}
	\label{sec:lb_copy_task}
	To begin with, we consider a task that is related to the matrix vector product but seems somewhat simpler.
	In~\cite{TCS10}, a copy task is described, where a matrix in column major layout is created from a vector such that each non-zero element in row $i$ is a copy of the $i$th vector element.
	The lower bound for this task is quite similar to the bound for permuting in~\cite{av88}.
	In~\cite{TCS10}, the copy task is only used as a preliminary analyses to reduce the matrix vector product to.
	Here, we actually can use the task itself to state a lower bound.
	Observe that the copy task is equivalent to the creation of a sparse matrix in row major layout where each non-zero entry in column $j$ is a copy of the $j$th vector element.
	This corresponds to the special case of MapReduce where a map function simply copies its input value $H/N_M$ times with random intermediate indices which then have to be ordered by intermediate index.
	Thus, a lower bound for this task states a lower bound for the shuffle step with parallel map functions but non-parallel reduce.
	We extend this task further to $v$ multiple input vectors, such that in the created matrix each non-zero entry in column $j$ is a copy of of the $j$th vector element of one of the $v$ vector.
	In the following, we bound the number of I/Os required for a family of programs for the copy task such that every matrix conformation (position of the non-zero elements) can be created by a program.

	\paragraph{Normalisations}
	In the following, we make some normalisations to programs which will not increase the number of I/Os.
	Therefore, it suffices to consider normalised programs only.
	Computational operations are not required for the copy task and can hence be removed from the program with all their results.
	This can only reduce the number of I/Os.
	We further can remove any other element created during the program that does not run into the final matrix.
	Hence, input vector elements and any copies that do not run into the final result are removed throughout the whole program.
	In this, we also remove elements that are no long required in internal memory immediately after an I/O.
	Moreover, we eliminate copy operations within internal memory.
	Multiple copies of an input vector element in internal memory are useless, and copies can be generated by an output where elements reside in internal memory after the output.
	This step can again only reduce the number of elements in internal memory, and thus, the number of I/Os.

	\paragraph{Abstraction}
	We abstract like in~\cite{TCS10} from the actual configuration.
	Instead of the full information of an element, we consider only the set of \emph{(column index, origin vector index)}~tuples of the elements in a block or in internal memory.
	Hence, blocks and internal memory are considered subsets of $\{(1,1),\dots,(N_M,v)\}$ of size at most $B$ and $M$, respectively.
	This abstracts especially from the ordering of elements and the multiplicity of elements with the same column index.

	\paragraph{Description of programs}
	We consider the change of configurations over time.
	The initial abstract configuration is unique over all programs for fixed $N_M$.
	The final configuration in contrast depends on the conformation of the matrix that is generated.
	This number will be examined later on.

	For a given abstract configuration, we examine the number of possibly succeeding (abstract) configurations generated by different normalised programs.
	To this end, first assume that processor $i$ performs an input while all the other processors stay idle.
	After $\ell$ parallel I/Os, there can be at most $H/B + P \ell \leq 2 P \ell$ non-empty blocks in external memory.
	Hence, there are at most $2 P \ell$ blocks that can be read by the input.
	Afterwards, up to $B$ new elements are added to internal memory of processor $i$.
	We assume that unneeded elements are removed right away, which allows a choice of $2^B$ such that there are at most $2^B 2\ell$ succeeding configurations.
	Now consider the case processor $i$ performs an output.
	For $\ell$ being the maximal number of I/Os, there are less than $4 P \ell$ positions available relative to the non-empty blocks to perform the output to: $2 P (\ell - 1)$ non-empty blocks that can be overwritten and another $2 P (\ell - 1) + 1$ empty positions relative to the non-empty blocks.
	The content of the output block can consist of up to $B$ out of the $X_{i,l}$ elements in internal memory of processor $i$.
	Further unneeded elements can be removed after the output which constitutes another $2^B$ different possible configurations.

	Each of the $P$ processors can perform either an input, an output or be idle during the $l$th parallel I/O.
	Thus, there are up to $3^P \prod_{i=1}^P {X_{i,l} + B \choose B} \cdot 2^B \cdot 4 P \ell$ possible configurations succeeding a given configuration.
	A family of normalised programs with $\ell$ parallel I/Os can lead to
	\begin{equation}
		\label{eq:trace_forward}
		\prod_{l=1}^\ell 3^P \prod_{i=1}^P {X_{i,l} + B \choose B} 2^B \cdot 4 P \ell
	\end{equation}
	distinct abstract configurations.

	\paragraph{Different abstract matrix conformations}
	For a family of programs being able to produce all conformations, $\ell$ needs to be large enough such that (\ref{eq:trace_forward}) is at least as large as the number of abstract configurations representing all conformations.
	There are $\binom{N_M}{H/N_R}^{N_R}$ different conformations of $N_R \times N_M$ matrices with $H/N_R$ non-zero entries per row.
	Furthermore, each of the non-zero entries can stem from one of the $v$ input vectors.
	However, since we consider abstract configurations and ignore the ordering and multiplicity of elements within a block, the number of final abstract configurations is less.
	For an abstract configuration, it is not clear wether a tuple present in a block stems from one or multiple rows, neither from which of them.
	As described in~\cite{TCS10} the following cases have to be distinguished.
	If $H/N_R = B$ each block corresponds to exactly one row such that each abstract configuration describes only a single conformation.
	In case $H/N_R > B$, a block contains entries from at most two rows.
	Hence, a column index in the abstract description of a block can origin either from the first, second or both rows.
	For $H/N_R < B$, a block contains entries from $BN_R/H$ different rows.
	Then, the at most $B$ indices in the abstract description can be column indices of each of the $BN_R/H$ rows.
	Thus, there are ${B \choose H/N_R}$\todo{ohne 2?} blocks with the same abstract description.

	\begin{theorem}
		\label{thm:copy_task}
		Given block size $B$, internal memory $M \geq 3B$, number of mappers $N_M \geq 9^{1/\epsilon}$ and the number of processors $P \leq \frac{H}{B}$.
		Creating a sparse $N_R \times N_M$ matrix with $H$ non-zero entries in row major layout from $v \leq H/N_M$ vectors such that $a_{ij} = x_j^{(k)}$ for a $1 \leq k \leq v$ for each non-zero entry has (parallel) I/O-complexity
		\[ \ell \geq \mins{\frac{\epsilon^2}{5}\frac{H}{P}, \frac{H}{7PB} \log_{\mins{\frac{M}{B}, \frac{2H}{PB}}} \mins{\frac{N_M N_R v}{3H}, \frac{N_M v}{eB}}} \text{\,.} \]
	\end{theorem}
	\begin{proof}
		If a family of programs with $\ell$ I/Os is able to create all conformations of $N_R \times N_M$ matrices in row major layout with $H$ non-zero entries, then
		\begin{equation}
			\label{eq:row_lb}
			\prod_{l=1}^\ell 3^P \prod_{i=1}^P {X_{i,l} + B \choose B} 2^B \cdot 4 P \ell \geq {N_R \choose H/N_M}^{N_M} v^H / \tau_R
		\end{equation}
		with
		\begin{equation*}
			\tau_R \leq \begin{cases}
				3^H          & \text{if $B < H/N_R$}\\
				1            & \text{if $B = H/N_R$}\\
				(eBN_R/H)^H  & \text{if $B > H/N_R$}
			\end{cases}
		\end{equation*}
		has to hold.
		For $\ell \geq \frac{1}{5} \frac{H}{P}$, the claim is already proven, so we can assume in the following $\ell < \frac{1}{4} \frac{H}{P}$.
		Similar, if $N_M \leq 2^8$, then the logarithm in Theorem~\ref{thm:copy_task} is smaller than $7$ such that the theorem is proven with a scanning bound of $\frac{H}{PB}$ which is required for reading the input.
		Hence, also assume $N_M > 2^8$.
		Taking logarithms and estimating binomial coefficients in (\ref{eq:row_lb}) yields
		\[  \ell P (B+\log 3H) + \sum_{l=1}^\ell \sum_{i=1}^P B \log \frac{e(X_{i,l} + B)}{B} \geq H \log \frac{N_M N_R}{H} + H \log v - \log \tau_R \text{.} \]
		Observe that for all $l$ the term $\prod_{i=1}^P X_{i,l}$ with $\sum_{i=0}^P X_{i,l} \leq \mins{PM, H}$ is maximised when $X_{1,l} = \dots = X_{P,l} = \mins{M, H/P}$.
		Hence, we have
		\[  \ell P \left(B + \log 3H + B \log e \frac{\mins{M, H/P} + 1}{B} \right) \geq H \log \frac{N_M N_R v}{H} - \log \tau_R  \text{\,.}\]
		After substituting $\tau$ and isolating $\ell$, we obtain
		\[ \ell \geq \frac{H}{P} \frac{\log \mins{\frac{N_M N_R v}{3H}, \frac{N_M v}{eB}}}{\log 3H + B \log \left(2e (\mins{\frac{M}{B}, \frac{H}{PB}} + 1)\right)} \]
		and using the assumptions $M \geq 3B$ and $H/(PB) \geq 1$, we get
		\[ \ell \geq \frac{H}{P} \frac{\log \mins{\frac{N_M N_R v}{3H}, \frac{N_M v}{eB}}}{\log 3H + \frac72 B \log \mins{\frac{M}{B}, \frac{2H}{PB}} } \text{.}\]
		where we used that $2e(x+1) \leq (2x)^{7/2}$ for $x \geq 1$.
		\\
		\emph{Case 1:}
		For $\log 3H \leq \frac72  B \log \mins{\frac{M}{B}, \frac{2H}{PB}}$, we have
		\[ \ell \geq \frac{H}{7PB} \log_{\mins{\frac{M}{B}, \frac{2H}{PB}}} \mins{\frac{N_M N_R v}{3H}, \frac{N_M v}{eB}} \text{.}\]
		\\
		\emph{Case 2:}
		If $\log 3H \geq \frac72 B \log \mins{\frac{M}{B}, \frac{2H}{PB}} $, we can use $3H \leq H^{5/4}$ for $H \geq N_M \geq 2^8$, and with Lemma~\ref{lem:bounds}.i, we obtain
		\[ \ell \geq \frac{H}{P} \frac{\log \mins{\frac{N_M N_R v}{3H}, \frac{N_M v}{eB}}}{\frac52 \log N_M^{1/\epsilon}} \text{\,.}\]
		Using Lemma~\ref{lem:bounds}.ii with $N_1 = N_R$ and $N_2 = N_M$, and ignoring $v$, we have
		\[ \ell \geq \frac{H}{P} \frac{\log \mins{\frac13 N_M^\epsilon, \epsilon N_M^{5/8}}}{\frac52 \log N_M^{1/\epsilon}} \text{\,.} \]
		For $N_M \geq 3^{2/\epsilon}$ in term, we obtain
		\[ \ell \geq \frac{H}{P} \frac{\log \mins{N_M^{\epsilon/2}, \epsilon N_M^{5/8}}}{\frac52 \log N_M^{1/\epsilon}} \geq \frac{H}{P} \frac{\log N_M^{\epsilon/2}}{\frac52 \log N_M^{1/\epsilon}} \geq \frac{\epsilon^2}{5}\frac{H}{P}\]
		since $\epsilon N_M^{5/8} \geq N_M^{\epsilon/2}$ for $N_M \geq 1/\epsilon^{\frac{8}{5-4\epsilon}}$ and $9^{1/\epsilon} \geq 1/\epsilon^{\frac{8}{5-4\epsilon}}$ for all $\epsilon > 0$.
	\end{proof}

	\subsection*{Mixed column layout with multiple output pairs}
	\label{sec:lb_worstcase}
	Following~\cite{TCS10}, for this task, it suffices to consider the multiplication of a matrix with the all-ones-vector, i.e., the task of simply building row sums of the matrix.
	Similarly, we extend the task such that each elementary product is only used for the calculations of one of the $w$ output vectors.
	Hence, we consider $w$ independent tasks of building subsets of row sums.
	This task can be seen as a time-inverse variant of the copy task described in Section~\ref{sec:lb_copy_task}.
	Instead of spreading copies of input elements, the scattered matrix elements of the same row have to be collected and summed up to $w$ subset sums.
	Therefore, we can use a similar analysis as before, but consider the change of configurations backwards in time since there are multiple input forms depending on the conformation that can all create the same output.

	\paragraph{Normalisation}
	Again, we describe some normalisations to programs for this task that do not increase the required number of I/Os.
	First, observe that each non-zero entry can only appear once in the final result since we do not guarantee inverse elements.
	Hence, copies of matrix entries can be removed from the program.
	Second, we assume that non-zero entries from the same row are summed immediately when present in internal memory.
	This can only decrease the number of elements in internal memory at a time, and hence the number of I/Os.
	Finally, an element in internal memory shall be deleted immediately after an I/O if it is no longer used.
	This implies that an element that is output will be removed from internal memory immediately after its output since we argued before that copying is useless.
	Furthermore, when an input is performed but some of the elements are removed immediately after the input, one can think of loading only the elements that are actually required.

	\paragraph{Abstraction}
	We consider an abstract configuration similar to the section before.
	Here we consider the set tuples of \emph{row} indices and \emph{destinations} of the elements (or sums of elements).
	Note that sums can only consist of elements from the same row and the same destination, otherwise the created sum cannot run into the final result.
	Hence, internal memory and each block states a subset of $\{(1,1),\dots,(N_R,w)\}$ of size up to $M$ and $B$, respectively.

	\paragraph{Description of programs}
	The proof for this task considers the change of configurations backwards in time from the final configuration to the initial.
	Because concurrent reads are possible in our model, we assume that elements in a block in external memory, that do not belong to the final output disintegrate magically, when read for the last time.
	This holds especially for the case an output is performed to the $i$th block on disk.
	Since this will replace the elements present before, the former elements cannot be read any more and hence, disintegrated before.
	Thus, outputs are only performed to empty blocks.

	Consider the final configuration after $\ell$ I/Os.
	Since all blocks that do not belong to the output disintegrated, the final (abstract) configuration is unique for all programs that compute the matrix vector product for fixed $N_R$.
	In contrast, the initial configuration depends on the conformation of the matrix, the position of the non-zero elements.
	For a family of programs that create the matrix vector product for sparse $N_R \times N_M$ matrices, the changes of configurations can be consider as a tree rooted in the final configuration.
	Each leaf corresponds to a different matrix conformation and each layer of depth $i$ corresponds to the configurations at time $\ell - i$.

	%

	%
	We normalised our programs such that computational operations are only done immediately after an I/O.
	Hence, it suffices to consider input and output operations only.
	Given a certain configuration after $l$ I/Os, we now need to count the number of \emph{possible} preceding configurations.
	Note that since we normalised our programs to avoid copying matrix entries, at any point in time there can be no more than $H$ non-empty blocks on disk.

	For the $l$th parallel I/O, consider the case processor $i$ performs an input.
	The $l$th configuration can stem from several preceding configurations.
	To upper bound the number of preceding configurations, we assume that some elements of the input block disintegrated immediately after the input.
	This changes the configuration on disk, otherwise only the change of internal memory has to be considered.
	Again, there are at most $H/B + P \ell \leq 2 P \ell$ non-empty positions on disk to choose which block was input.
	Furthermore, there are at most $B$ row indices input.
	For $X_{i,l}$ being the number of distinct row indices in internal memory after the $l$th I/O, there are at most ${X_{i,l} + B \choose B}$ possibilities which row indices belong to the input.
	Each of which could however have been present before, or appear as a new index.
	This makes another $2^B$ possibilities.
	Altogether, there are up to ${X_{i,l} + B \choose B} 2^B 2P\ell$ preceding configurations if processor $i$ performs an input, while all other processors are idle.

	Now, consider the case that processor $i$ performs an output.
	This alters one block on disk that was empty before.
	Hence, there are up to $4 P \ell$ possible preceding configurations.
	In the preceding configuration, the elements of the output block are in internal memory.
	However, when choosing the block position where the output is performed, these elements are determined by the (known) configuration after the output.

	Each of the $P$ processors can perform either an input, an output or be idle during the $l$th I/O.
	Thus, there are up to $3^P \prod_{i=1}^P {X_{i,l} + B \choose B} 2^B 4 P \ell$ possible preceding configurations.
	A family of programs with $\ell$ parallel I/Os can hence create the matrix vector product from at most
	\begin{equation}
		\label{eq:trace}
		\prod_{l=1}^\ell 3^P \prod_{i=1}^P {X_{i,l} + B \choose B} 2^B 4 P \ell
	\end{equation}
	initial abstract configurations.
	This result is still independent from the layout of the matrix.

	\paragraph{Different abstract matrix conformations}
	To gain a lower bound on the I/O-complexity of the mixed column layout, we have to lower bound (\ref{eq:trace}) by the number of different matrix conformations in mixed column layout expressed in abstract blocks on disk.
	We consider matrices with exactly $H/N_M$ elements per column.
	Think of drawing the non-zero elements for each column one after another.
	For the number of non-zero entries per column $H/N_M \leq N_R/2$, there are at least $(N_R/2)$ possibilities to draw the position of a non-zero element.
	Furthermore, each element can be involved in the computation of one of the $w$ output vectors.
	In total, there are at least $(N_R/2)^H w^H$ different matrix conformations.
	However, abstracting to the view of \emph{(row index, destination)} tuples, the ordering of elements within a block gets hidden.
	Additionally, if a block contains elements from several columns, in its abstraction it is not clear from which column(s) a tuple may stem from.
	The number of different conformation that correspond to the same abstract conformation can be bounded from above by $B^H$ since each element can be one of the at most $B$ row indices of its block.

	\begin{theorem}
		\label{thm:mixed_column}
		Given block size $B$, internal memory $M \geq 3B$, number of reducers $N_R \geq 1/\epsilon^{8/3}$ and the number of processors $P \leq \frac{H}{B}$.
		Creating the combined matrix vector product for a sparse $N_R \times N_M$ matrix in mixed column layout with $H$ non-zero entries for $w \leq H/N_R$ output vectors, with $H/N_R \leq N_M^{1-\epsilon}$ and $H/N_M \leq N_R^{1-\epsilon}$ for $\epsilon > 0$ has (parallel) I/O-complexity
		\[ \ell \geq \mins{\frac{\epsilon}{10} \frac{H}{P}, \frac{H}{7 PB} \log_{\mins{\frac{M}{B}, \frac{2H}{PB}}} \frac{N_Rw}{2B} } \text{\,.} \] 
	\end{theorem}
	\begin{proof}
		A lower bound on the minimal number of I/Os $\ell$ required for a family of programs that create the matrix vector product for $N_R \times N_M$ matrices with $H$ entries in mixed column layout is given by
		\begin{equation}
			\prod_{l=1}^\ell 3^P \prod_{i=1}^P {X_{i,l} + B \choose B} 2^B 4 P \ell \geq \kfrac{N_R}{2B}^H w^H \text{\,.}
		\end{equation}
		With similar arguments as in the proof of Theorem~\ref{thm:copy_task}, we can assume $\ell < \frac14 H/P$, and $N_R \geq 2^8$. Otherwise the theorem holds trivially.
		Taking logarithms and estimating binomial coefficients yields
		\[ \ell P (B + \log 3H) + \sum_{l=1}^\ell \sum_{i=1}^P B \log \frac{eX_{i,l} + eB}{B} \geq H \log \frac{N_R}{2B} + H \log w  \]
		As described in Section~\ref{sec:lb_copy_task}, the left-hand side is maximised for each $l$ when $X_{1,l} = \dots = X_{P,l} = \mins{M, H/P}$.
		Thus, we obtain
		\[ \ell P \left(B + \log 3H + B \log \frac{e\mins{M, H/P} + eB}{B}\right) \geq H \log \frac{N_Rw}{2B}  \]
		and isolating $\ell$, we can estimate
		\[ \ell \geq \frac{H}{P} \frac{\log \frac{N_Rw}{2B}}{\log 3H + \frac72 B \log \mins{\frac{M}{B}, \frac{2H}{PB}}} \text{.} \]
		where we used again $M \geq 3B$ and $H/ (PB) \geq 1$.

		\emph{Case 1:} For $\log 3H \leq \frac72 B \log \mins{\frac{M}{B}, \frac{2H}{PB}}$, the lower bound matches the sorting algorithm:
		\[ \ell \geq \frac{H}{7 PB} \log_{\mins{\frac{M}{B}, \frac{2H}{PB}}} \frac{N_Rw}{2B} \]

		\emph{Case 2:} For $\log 3H > \frac72 B \log \mins{\frac{M}{B}, \frac{2H}{PB}}$, we get
		\[ \ell \geq \frac{H}{P} \frac{\log \frac{N_Rw}{2B}}{2 \log 3H} \text{.}\]
		Using Lemma~\ref{lem:bounds}.i, and $3H \geq H^{5/4}$ for $H \geq 2^8$, we have
		\[ \ell \geq \frac{H}{P} \frac{\log \frac{N_Rw}{2B}}{\frac52 \log N_R^{1/\epsilon}} \]
		and with Lemma~\ref{lem:bounds}.ii, we get
		\[ \ell \geq \frac{H}{P} \frac{\log \epsilon N_R^{5/8}w}{\frac52 \log N_R^{1/\epsilon}} \geq \frac{\epsilon}{10} \frac{H}{P}\]
		for $N_R \geq 1/\epsilon^{8/3}$ which is matched by the direct algorithm.
	\end{proof}

	\subsection*{Column major layout with multiple output pairs}
	\label{sec:lb_column}
	For column major layout, the number of different abstract matrix conformations corresponds to the number of abstract conformation described in Section~\ref{sec:lb_copy_task} but with $N_M$ and $N_R$ exchanged.


	If a family of programs with $\ell$ I/Os is able to create the matrix vector product with each $N_R \times N_M$ matrix with $H$ non-zero entries to obtain $w \leq H/N_R$ vectors of row sum subsets, then
	\begin{equation}
		\label{eq:column_lb}
		\prod_{l=1}^\ell 3^P \prod_{i=1}^P {X_{i,l} + B \choose B} 2^B 4P\ell \geq {N_M \choose H/N_R}^{N_R} w^H / \tau_M
	\end{equation}
	with
	\begin{equation*}
		\tau_M \leq \begin{cases}
			3^H          & \text{if $B < H/N_M$}\\
			1            & \text{if $B = H/N_M$}\\
			(eBN_M/H)^H  & \text{if $B > H/N_M$}
		\end{cases}
	\end{equation*}
	has to hold.
	This yields the following theorem.
	
	\begin{theorem}
		\label{thm:column_major}
		Given block size $B$, internal memory $M \geq 3B$, number of mappers $N_M \geq 9^{1/\epsilon}$ and the number of processors $P \leq \frac{H}{B}$.
		Creating the combined matrix vector product for a sparse $N_R \times N_M$ matrix in column major layout with $H$ non-zero entries for $w$ output vectors, and $H/N_R \leq N_M^{1-\epsilon}$ and $H/N_M \leq N_R^{1-\epsilon}$ for $\epsilon > 0$ has (parallel) I/O-complexity
		\[ \ell \geq \mins{\frac{\epsilon^2}{5}\frac{H}{P}, \frac{H}{7PB} \log_{\mins{\frac{M}{B}, \frac{2H}{PB}}} \mins{\frac{N_M N_R w}{3H}, \frac{N_R w}{eB}}} \text{\,.} \]
	\end{theorem}

	\subsection*{Best-case layout with multiple input and output pairs}
	\label{sec:lb_best}
	For the best-case layout, the algorithm is allowed to choose the layout of the matrix.
	This makes the task of building row sums become trivial by setting the layout of the matrix to row major layout.
	Hence, we have to follow the movement and copying of input vector elements as well.

	For this task, we consider both, multiple input and multiple output vectors.
	Recall that any intermediate pair stems from exactly one input element, and it is required for the computation of only a single output element.
	For each of the $H$ non-zero entries in our matrix there is not only a choice of its position but also the choice from which of the $v$ element it stems from and which of the $w$ elements is its destination.
	This results in a total number of $\binom{N_M N_R}{H} v^H w^H$ different tasks.
	We assume that $v \geq H/N_M$ and $w \geq H/N_R$ such that all input and output elements can be useful.

	The task of creating this type of matrix vector product can be seen as spreading input elements to create elementary products $a_{ij} x^{(k_{ij})}_j$, and then collecting these elementary products to form the output vectors.
	To describe these two main tasks, we distinguish between matrix entries $a_{i,j}$, elementary products $a_{ij} x^{(k_{ij})}_j$ and partial sums $\sum_{j \in \mathcal{S}} a_{ij} x^{(k_{ij})}_j$ which we trace as described in Section~\ref{sec:lb_column}, and input vector elements which will be followed as described below.
	To the first group (matrix entries, elementary products, partial sums), we also refer as row elements and we call the row index their index.

	\paragraph{Normalisation}
	At first, we make again use of the normalisations described already.
	In doing so, we distinguish time between input vector elements and row elements.
	For row elements, we apply the normalisations described in Section~\ref{sec:lb_worstcase} (multiple copies are removed, row elements are summed immediately and unused elements are removed from internal memory after an I/O).
	For input vector elements, we use the normalisations of Section~\ref{sec:lb_copy_task}.
	Furthermore, we create an elementary product $a_{ij} x^{(k_{ij})}_j$ as soon as both elements reside in memory for the first time, i.e. immediately after an input.
	The elementary product is stored at the internal memory position where $a_{ij}$ was present.
	This is possible since the variable $a_{ij}$ is not required for any other operation than a multiplication with $x^{(k_{ij})}_j$, and can hence be replaced.
	This normalisation does not increase the memory usage at any time.

	\paragraph{Abstraction}
	As the sections before and analogue to~\cite{TCS10}, we consider an abstraction of the current configuration.
	In this, we define the set of tuples $\mathcal{M}^V_{i,j} \subseteq \{(1,1),\dots,(N_M,v)\}$, $|\mathcal{M}^V_{i,j}| \leq M$ of input vector elements that are in internal memory of processor $i$ after the $j$th parallel I/O.
	Similar, let $\mathcal{M}^R_{i,j} \subseteq \{(1,1),\dots,(N_R,w)\}$, $|\mathcal{M}^R_{i,j}| \leq M$ be the set of tuples of row elements residing in internal memory of processor $i$ after the $j$th parallel I/O.
	Additionally, for each block on disk, we consider the set of tuples of input vector elements contained and call this over all blocks on disk together with $\mathcal{M}^V_{1,j}, \dots \mathcal{M}^V_{P,j}$ the abstract input configuration after the $j$th parallel I/O.
	Analogously, we define the abstract row configuration after the $j$th parallel I/O to be $\mathcal{M}^R_{1,j}, \dots, \mathcal{M}^R_{P,j}$ together with the sets of tuples of row elements of each block on disk.

	\paragraph{Description of programs}
	The final abstract row configuration is again unique over all programs that create the matrix vector product for fixed $N_R$.
	Hence, we can apply our analysis from Section~\ref{sec:lb_column} to describe the number of different abstract row configurations that can be present during the execution of a program with $\ell$ I/Os after the $l$th I/O.
	To describe the abstract input configurations, we can use the analysis from Section~\ref{sec:lb_copy_task}.
	Together, the number of distinct (input vector / row) configurations that can result from a given (initial / final) configuration after $\ell$ parallel I/O is bounded by~(\ref{eq:trace}).

	So far, we described all operations concerning only input vector elements, or only row elements, but no interaction between these elements.
	There is only one operation that performs an interaction, and that is multiplying an input element with a non-zero entries which creates a row element.
	For each non-zero entry $a_{ij}$, exactly one multiplication with a $x^{(k_{ij})}_j$ is performed during the whole execution of the program.
	Hence, if in our abstraction, a row element with index $i$ was created by multiplication from an input variable with index $j$, this fixes $a_{ij}$ to be non-zero.
	We normalised programs such that an elementary product is created immediately when both elements appear in internal memory together for the first time.
	This can only happen after an input.
	Thus, for each input there are at most $B$ new elements in internal memory and some $X_{i,l}$ elements that are already in internal memory.
	Hence, there are at most $B X_{i,l}$ possibilities where a multiplication can be done for each processor after each I/O.
	In total, we get $Y = \sum_{l=1}^\ell \sum_{i=1}^P B X_{i,l} \leq \ell B \cdot \mins{PM, H}$ possibilities where a multiplication can be performed.
	Together with the traces that describe the movement and copying / summing of input and row elements, this give a unique description of the abstract matrix conformation.

	\begin{theorem}
		Given block size $B$, internal memory $M \geq 3B$, and the number of processors $P \leq \frac{H}{B}$.
		Creating the combined matrix vector product for a sparse $N_R \times N_M$ matrix in best-case layout with $H$ non-zero entries for $v \leq H/N_M$ input and $w \leq H/N_R$ output vectors, with $H/N_R \leq N_M^{1/6}$ and $H/N_M \leq N_R^{1/6}$ has (parallel) I/O-complexity
		\[ \ell \geq \mins{\frac{1}{20} \frac{H}{P}, \frac{H}{14 PB} \log_{\mins{\frac{M}{B}, \frac{2H}{PB}}} \frac{N_M N_R v w}{H \mins{M, H/P}} } \text{\,.} \] 
	\end{theorem}
	\begin{proof}
		With the above observations, for a family of programs with $\ell$ I/Os that create the matrix vector product for any $N_R \times N_M$ matrix with $H$ non-zero entries where the layout of the matrix can be chosen by the program it has to hold
		\[ {Y \choose H} \cdot \prod_{l=1}^\ell 3^P \prod_{i=1}^{P} {X_{i,l} + B \choose B} 2^B 4 P \ell \geq {N_R \choose H/N_M}^{N_M} v^H w^H \text{.}\]
		Like in the proofs before, we argue that $\ell < \frac{H}{4P}$ and $N_R \geq 2^8$, otherwise the claim holds trivially.
		The left-hand side is again maximised for $X_{i,l} = \mins{M, H/P}$ for all $i, l$.
		Estimating binomials and taking logarithms, we have
		\[ \ell P \left(\log 3H + \frac72 B \log \mins{\frac{M}{B}, \frac{2H}{PB}} \right) \geq H \left( \log \frac{N_M N_R vw}{H} - \log \frac{e \ell B \mins{PM, H}}{H} \right) \]
		where we followed the calculations given for the other layouts. 
		Reordering terms, we obtain
		\[ \ell \geq \frac{H}{P} \frac{\log \frac{N_M N_R v w}{e \ell B \mins{PM, H}}}{\log 3H + \frac72 B \log d}\]
		with $d = \mins{M/B, 2H/(PB)}$.
		By Lemma A.2 in~\cite{TCS10}, $x \geq \frac{\log_b(s/x)}{t}$ implies $x \geq \frac{\log_b (s\cdot t)}{2t}$.
		For $x=\ell$, $s = \frac{N_M N_R}{e B \mins{PM, H}}$ and $t = \frac{P}{H} \left(\log 3H + \frac72 B \log d \right)$ this results in 
		\[ \ell \geq \frac{H}{2P} \frac{\log \frac{N_M N_R v w P \left(\log 3H + \frac72 B \log d\right)}{e B H \mins{PM, H}}}{\log 3H + \frac72 B \log d} \text{.}\]
		\emph{Case 1:}
		For $\log 3H \leq \frac72B \log d$, we get a lower bound of
		\[ \ell \geq \frac{H}{14P} \frac{\log \frac{N_M N_R v w}{H \mins{M, H/P}}}{B \log d} \text{.}\]
		\emph{Case 2:}
		For $\log 3H > \frac72B \log d$, we get
		\[ \ell \geq \frac{H}{4P} \frac{\log \frac{N_M N_R v w\log 3H}{e B H \mins{M, H/P}}}{\log 3H } \geq \frac{H}{4P} \frac{\log \frac{N_M N_R}{e H B d}}{\frac54 \log H }\]
		for $H \geq 2^8$.
		Setting $H/N_M \leq N_R^\frac16$ and using Lemma~\ref{lem:bounds}.i with $\epsilon = \frac56$, we get 
		\[ \ell \geq \frac{H}{5P} \frac{\log \frac{N_R^{5/6}}{e B d}}{\log N_R^{6/5} } \]
		for $N_R \geq 9$.
		Using Lemma~\ref{lem:bounds}.iii, with $N_1 = N_M$ and $N_2 = N_R$, we get
		\[ \ell \geq \frac{H}{6P} \frac{\log \frac{\frac56 N_R^{5/6}}{N_R^{3/8} N_R^{3/(7B)}}}{\log N_R } > \frac{H}{6P} \frac{\log \frac56 N_R^{1/3}}{\log N_R } \geq \frac{H}{20 P} \]
		for $B \geq 4$, using $N_R \geq 2^8 > (\frac65)^{30}$ such that $\frac56 N^{1/3} \geq N^{3/10}$.
	\end{proof}

	\subsection*{Transposing bound}
	Another method to obtain a lower bounds for the I/O-complexity is presented in~\cite{av88}.
	They use a potential function to lower bound the complexity of dense matrix transposition.
	This bound can also be applied to sparse matrix transposition, if the matrix is given in column major layout.

	In the following, we extend the potential to the PEM model.
	For each block $j$ written on disk at time $t$, its togetherness rating is defined by
	\[ \mathcal{B}_j(t) = \sum_{i=1}^{H/B} f(x_{ij}(t)) \]
	where $x_{ij}(t)$ is the number of elements present in block $j$ at time $t$ that belong to the $i$th output block, and
	\[ f(x) = \begin{cases}
		 x \log x & \text{for $x > 0$} \\
		 0 & \text{otw.}
	\end{cases}\]
	Similarly, to internal memory of processor $k$, a togetherness rating of
	\[ \mathcal{M}_k(t) = \sum_{i=1}^{H/B} f(y_{ik}(t)) \]
	is assigned with $y_{ik}(t)$ being the number of elements that belong to output block $i$ and reside at time $t$ in internal memory of processor $k$.
	The potential is then defined by
	\[ \Phi(t) = \sum_{k=1}^P \mathcal{M}_k(t) + \sum_{j=1}^\infty \mathcal{B}_j(t) \text{\,.} \]
	In~\cite{av88}, it is shown that the change of the potential induced by an I/O is $\Delta\Phi(t) < 2B \log \frac{M}{B}$.
	To extend the bound for the PEM-model, we have to restate this argument for parallel I/Os.
	\begin{lemma}
		The increase of the potential during one parallel I/O is bounded by
		\[ \Delta\Phi(t + 1) \leq PB \log 2e + PB \log \frac{\mins{M, H/P}}{B} \text{\,.} \]
	\end{lemma}
	\begin{proof}
		Obviously, an output can never increase the potential.
		An input of block $j$ by processor $k$ in contrast induces the following change in the potential
		\[ \Delta\Phi(t + 1) = \sum_{i=1}^{H/B} f(y_{ik}(t) + x_{ij}(t)) - f(y_{ik}(t)) - f(x_{ij}(t)) \text{\,.} \] 
		Hence the maximal increase of the potential function is bounded by
		\[ \Delta\Phi(t + 1) \leq \sum_{k=1}^P \sum_{i=1}^{H/B} f(y_{ik}(t) + x_{ij_k}(t)) - f(y_{ik}(t)) - f(x_{ij_k}(t)) \text{\,.} \]
		where $j_k$ is the index of the block input by processor $k$.

		Let $I_{kj_k}(t)$ be the set of indices $i$ such that $x_{ij_k}(t) \geq 1$ and $y_{ik}(t) \geq 1$.
		Substituting $f(x)$, we obtain
		\begin{eqnarray*}
			\Delta\Phi(t + 1) & \leq & \sum_{k=1}^P \sum_{i \in I_{kj_k}(t)} y_{ik}(t) \log \frac{y_{ik}(t) + x_{ij_k}(t)}{y_{ik}(t)} + x_{ij_k}(t) \log \frac{y_{ik}(t) + x_{ij_k}(t)}{x_{ij_k}(t)} \\
				& \leq & \sum_{k=1}^P \sum_{i \in I_{kj_k}(t)} \frac{x_{ij_k}(t)}{\ln 2} + x_{ij}(t) \log \frac{y_{ik}(t) + x_{ij_k}(t)}{x_{ij_k}(t)} \\
				& \leq & PB \log e + \sum_{k=1}^P \sum_{i \in I_{kj_k}(t)} x_{ij_k}(t) \log \frac{2 y_{ik}(t)}{x_{ij_k}(t)} \\
				& \leq & PB \log 2e + \sum_{k=1}^P \sum_{i \in I_{kj_k}(t)} x_{ij_k}(t) \log \frac{y_{ik}(t)}{x_{ij_k}(t)} \text{\,.}
		\end{eqnarray*}

		Now let $\sum_{k=1}^P \sum_{i \in I_{kj_k}(t)} x_{ij_k}(t) = X(t)$ and $\sum_{k=1}^P \sum_{i \in I_{kj_k}(t)} y_{ij}(t) = Y(t)$ for $1 \leq k \leq P$.
		In order to upper bound $\Delta\Phi(t + 1)$, we substitute $\hat{x}_{ij_k}(t) = x_{ij_k}(t) / X(t)$ and $\hat{y}_{ik}(t) = y_{ik}(t) / Y(t)$ such that we get
		\begin{eqnarray*}
			\Delta\Phi(t + 1) & \leq & PB \log 2e + \sum_{k=1}^P \sum_{i \in I_{kj_k}(t)} X(t) \hat{x}_{ij_k}(t) \left[ \log \frac{\hat{y}_{ik}(t)}{\hat{x}_{ij_k}(t)} + \log \frac{Y(t)}{X(t)} \right] \nonumber \\
				& = & PB \log 2e + X(t) \log \frac{Y(t)}{X(t)} + X(t) \sum_{k=1}^P \sum_{i \in I_{kj_k}(t)} \kern-.8em \hat{x}_{ij_k}(t) \log \frac{\hat{y}_{ik}(t)}{\hat{x}_{ij_k}(t)} \text{\,.}
		\end{eqnarray*}

		The last sum describes a negative Kullback–Leibler divergence.
		Note that for fixed $k$ and $t$, $\hat{x}_{ij_k}(t)$ and $\hat{y}_{ik}(t)$ constitute a probability distribution over $\{1,\dots,k\} \times I_{kj_k}(t)$.
		It is well known that the Kullback–Leibler divergence is minimised when both probability distributions equal, and positive otherwise.
		Hence, since it appears in a negative form, the last term is upper bounded by $0$.

		Considering the second term, we obviously have $X(t) \leq PB$ and $Y(t) \leq \mins{PM, H}$.
		Thus, $X(t) \log \frac{Y(t)}{X(t)} \leq X(t) \log \frac{\mins{PM, H}}{X(t)}$ which is in term upper bounded by $B \log \mins{\frac{M}{B}, \frac{H}{PB}}$ since $\mins{PM, H} \geq B$.
	\end{proof}

	Since we can exclude copy operations for this task, the final potential is obviously $\Phi(\ell) = H \log B$. The initial potential in contrast has to be considered for each matrix separately.
	W.l.o.g., in the following let $N_M \geq N_R$.
	For $H/N_M \geq B$, any matrix that is row-wise $H/N_M$-regular and column-wise $H/N_R$-regular has an initial potential of $\Phi(0) = 0$:
	Observe that in this case each input block intersects with an output block in at most one element.
	Hence, the potential yields a lower bound of $\W{\frac{H}{PB} \log_{d} B}$ with $d = \mins{\frac{M}{B}, \frac{H}{PB}}$.

	For $H/N_M < B$, we consider the following matrix.
	All non-zero entry $a_{ij}$ are located at coordinates where $(i-1)H/N_R+1 \leq j \leq iH/N_R \mod N_M$ holds.
	Observe that any matrix with such a conformation stored in column or row major layout corresponds directly to a dense $h/N_M \times N_M$ matrix in column major layout, row major respectively.
	Thus, conform to \cite{av88} we get an initial potential of $\Phi(0) = h \log \maxs{1, B/\frac{h}{N_x}, B/N_x, B^2/h}$.
	Hence, we get for this dense matrix a lower bound of
	\begin{equation}
		\label{eq:transposing_bound}
		\W{\frac{H}{PB} \log_d \mins{B, \frac{H}{N_M}, N_M, \frac{H}{B}}} \text{\,.}
	\end{equation}

	\subsubsection{Combining the bounds}
	Consider a matrix given in column major layout.
	For the scenario $N_M > B > H/N_M$, the minimum breaks down to the term $H/N_M$.
	Combining the results from Theorem~\ref{thm:column_major} with the above bound we get $\W{\frac{H}{PB} \left(\log_d\frac{H}{N_M} + \log_d\frac{N_MN_R}{H} \right)}$ which is bound from below by $\W{\frac{H}{PB} \log_d N_R}$.
	Similar observations hold for $N_R > B > H/N_R$.
	Considering the other cases for the minimum in~(\ref{eq:transposing_bound}), we get a lower bound of
	\[ \W{\mins{ \frac{H}{P}, \frac{H}{PB} \log_d\mins{\frac{N_MN_RB}{H}, N_M, N_R, \frac{H}{B}}} } \]
	 I/Os for matrices in column major layout.
	
	Given a matrix in mixed column layout, we can apply~(\ref{eq:transposing_bound}) as well since column major is a special case of the mixed column layout.
	Hence, with similar considerations, we obtain a lower bound of
	\[ \W{\mins{\frac{H}{P}, \frac{H}{PB} \log_d\mins{N_R, \frac{H}{B}}} } \]
	 I/Os for matrices in mixed column layout.
\end{appendix}

\end{document}